\documentclass[11pt]{article}
\usepackage{graphicx}
\usepackage{wrapfig}
\usepackage{amsfonts}
\usepackage{amssymb,epic,eepic}
\usepackage{amsmath}
\usepackage{dcolumn}
\usepackage{bm}
\usepackage{bbm}
\usepackage{verbatim}
\usepackage{color}
\usepackage{amsthm}
\usepackage{tikz}
\usepackage{nameref}
\usetikzlibrary{arrows,positioning,decorations.pathmorphing,decorations.markings,shapes,fadings}
\usepackage[acronym]{glossaries}

\newcommand{\eins}{{\mathbbm{1}}}

\newcommand{\qa}{\Theta}
\newcommand{\Tr}{\mathrm{Tr}}

\newacronym{avc}{AVC}{Arbitrarily Varying Channel}
\newacronym{caavc}{CAAVC}{Correlation-Assisted Arbitrarily Varying Channel}
\newacronym{dmc}{DMC}{Discrete Memoryless Channel}
\newacronym{bpsk}{BPSK}{Binary Phase Shift Keying}
\newacronym{tmsv}{TMSV}{Two-Mode Squeezed Vacuum}
\newacronym{bsc}{BSC}{Binary Symmetric Channel}
\newacronym{cr}{CR}{Common Randomness}
\newacronym{dc}{DC}{Deterministic Code}
\newacronym{cc}{CC}{Correlated Code}
\newacronym{crac}{CRAC}{Common Randomness-Assisted Code}

\newtheorem{theorem}{Theorem}

\newtheorem{definition}{Definition}

\newtheorem{lemma}{Lemma}

\newtheorem{remark}{Remark}

\newcommand{\tr}{\mathrm{tr}}

\begin{document}
	\title{Transmitting Correlation for Data Transmission over the Bosonic Arbitrarily Varying Channel}
	\author{Janis N\"otzel, Florian Seitz}
	\maketitle
\begin{abstract}
	Shared randomness is the central ingredient for stabilizing symmetrizable communication systems against arbitrarily varying jammers. Given the presence of the jammer, however, the question arises how this precious resource could have been distributed. Several works discuss the use of external sources for this task. In this work, we show, based on the most standard optical communication model, how the sender and receiver can employ either classically correlated thermal light or entangled two-mode squeezed states created at and transmitted by the sender to counter the jamming attack of an energy-limited jammer during the distribution phase. Both sender and receiver are only allowed to use  homodyne detection in our model, and the sender has to obey a power limit as well.
\end{abstract}
	
\section{Introduction}
\glspl{avc} were introduced in \cite{bbt} and developed into the standard information-theoretic model for communication systems that are subject to jamming. The question whether the capacity $C$ of an \gls{avc} is positive or not turned out to be nontrivial \cite{ericsonAVC,elimination,csiszarNarayanPositivity}: As soon as an amount of randomness shared between sender and receiver is available that scales polynomially with the blocklength, the capacity $\bar C$ of the system is described by an information-theoretic quantity similar to Shannon's capacity formula. To describe $C$, the notion of \emph{symmetrizability} of an \gls{avc} got established \cite{ericsonAVC,csiszarNarayanPositivity} (see also \cite{ahlswede-blinovsky}). A symmetrizable \gls{avc} has zero capacity, and the capacity formula for a non-symmetrizable one equals that of the same \gls{avc} with shared randomness between sender and receiver. 

The investigation of \glspl{avc} in signal processing was initiated in \cite{csiszarNarayan}, where the received signal was modeled as $y= s+x+z$, with $s$ being the jamming signal, $x$ the transmitted signal, $z$ the additive white Gaussian noise ($s,x,z\in\mathbb C$). The analog of the model \cite{csiszarNarayan} in quantum signal processing is given by coherent states $|s\rangle,|x\rangle$ interfering on a beamsplitter with transmittivity $\eta$, with the received signal $|\sqrt{\eta}s+\sqrt{1-\eta}x+z\rangle$ being disturbed by classical Gaussian noise $z$ \cite{noetzel-isit2024}.

\glspl{avc} display super-activation \cite{minglaiSuperActivation} and discontinuity of $C$ \cite{minglaiDiscontinuity,noetzelDiscontinuity}. The symmetrizability condition is not computable \cite{bocheUncomputable}, but approximate solutions can be efficiently computed \cite{detectingsymmetrizability}. The presence of just any correlated bipartite source that is \emph{independent} from the jammer is sufficient to guarantee operation at $\bar C$ \cite{ahlswede-cai-correlated,boche13}. Their capacity under different types of jamming has been studied for colored Gaussian noise \cite{peregAVC} and native quantum tasks \cite{quantumAVC,fullyQuantumAVC}, with similar observations. In some cases, even exact thresholds for the amount of shared randomness are known \cite{jaggiAVC}. 

These results establish that shared (common) randomness is a critical ingredient for communication under jamming, and already extremely small amounts of shared randomness suffice to achieve the random capacity $\bar C$. 


However, the following core question has not been previously formulated: \emph{Assuming a symmetrizable \gls{avc}, and no preshared correlated randomness, how could sender and receiver start to communicate, how could they e.g. establish a source of shared randomness?}

In this work, we address this highly contradictory situation for the first time. We show that, under the constraint of an energy limit that applies per transmission, correlations arising from either specific classically correlated (non-entangled) quantum states or else pure quantum resources, namely shared entangled states. can be utilized to achieve positive capacity. Most importantly, \emph{the correlated states are created at the sender and shared with the receiver under the attack of the jammer}.

In order to model the situation where no correlated distributions are created and shared by the sender, we let sender and jammer use any displaced thermal state $S_N(\alpha)$ for the transmission, with the constraint $N+|\alpha|^2<E$ at the sender and $N+|\beta|^2<P$ at the jammer. Both signals get mixed at a $50:50$ beam-splitter for convenience, which is a standard quantum-optical communication model \cite{guhaBeamsplitter-1,guhaBeamsplitter-2}. For the case $P\geq E$, we show the channel has zero capacity.\\
Since the ability to create any shared common resource that would be helpful for data transmission by directly sending displaced thermal states $S_N(\alpha)$ with $|\alpha|^2+N\leq P$ would imply $C>0$, any such attempt must fail.\\ 
However, if we augment the transmitter with the ability to generate bipartite states $\rho_{RS}$ for which each reduced state $\rho_R$ satisfies $\rho_R=S_N(\alpha)$ for some $(\alpha,N)$ satisfying $|\alpha|^2+N\leq E\leq P$, the corresponding capacity can be proven to be positive. The class of states $\rho_{RS}$ which is allowed by our constraint includes both \gls{tmsv} states (which are entangled), but also specific two-mode thermal states (which are only classically correlated). In order to use the correlated states for data transmission the sender stores the $S$ system of $\rho_{SR}$ while transmitting the $R$ system. Then, both sender and receiver apply homodyne measurements and use the so-obtained results to create a correlated random variable. Afterwards, they can use this variable to obtain positive capacity by following the ideas outlined in \cite{ahlswede-cai-correlated}. 

\section{System Model and Result}
\paragraph{System Model}
\begin{definition}[{\gls{dc}}]\label{def:dc}
	An $(n,\lambda)$ \gls{dc} $\mathcal C$ allows only $\mathcal S=\{S_N(\alpha):N+|\alpha|^2\leq E\}$ for the sender and $\mathcal J=\{S_N(\beta):N+|\beta|^2\leq P\}$ at the jammer. It consists of a function $f:[M]\to\mathcal S^n$ as well as pairwise disjoint decoding sets $D_m\subset\mathbb R^n$, $m=1,\ldots,M$, such that 
	\begin{align}
		\mathrm{succ}(\mathcal C):=\min_{\mathbf j^n\in\mathcal J^n}\frac{1}{M}\sum_{m=1}^M\int_{D_m}\left[\prod_{i=1}^n\Tr|\hat x\rangle\langle\hat x|\mathcal N\left(S(\mathbf j^n)\otimes  S(f(m)\right)\right]dx^n\geq1-\lambda.
	\end{align}
	A non-negative number $R$ is called achievable with $\glspl{dc}$ if there exists a sequence of $\glspl{dc}$ such that $\liminf_{n\to\infty}\tfrac{1}{n}\log M_n\geq R$.
	The corresponding capacity $C$ is the supremum over all achievable numbers $R$.
\end{definition}
\begin{definition}[{\gls{crac}}]\label{def:crac}
	An $(n,M,\lambda)$ \gls{crac} $\mathcal C$ consists of a set $\{\mathcal C^{\gamma}\}_{\gamma=1}^\Gamma$ of $(n,\lambda^{(\gamma)})$ \glspl{dc}, where $\Gamma\in\mathbb N$. The success probability of a \gls{crac} is defined as 
	\begin{align}
		\mathrm{succ} =\min_{\mathbf j^n\in\mathcal J^n}\frac{1}{M\cdot\Gamma}\sum_{m,\gamma=1}^{M,\Gamma}\int_{D_m^\gamma}\left[\prod_{i=1}^n\Tr|\hat x\rangle\langle\hat x|\mathcal N\left(S(\mathbf j^n)\otimes  S(f^\gamma(m)\right)\right]dx^n\geq1-\lambda.
	\end{align}
	A non-negative number $R$ is called achievable with $\glspl{crac}$ if there exists a sequence of $\glspl{crac}$ such that $\liminf_{n\to\infty}\tfrac{1}{n}\log M_n\geq R$.
	The corresponding capacity $\bar C$ is the supremum over all numbers $R$ that are achievable wiht \glspl{crac}.
\end{definition}
\begin{remark}
	The numbers $\lambda^{(n)})$ in Definition \ref{def:crac} are mentioned for the sake of compliance with Definition \ref{def:dc}, but are themselves not relevant to Definition \ref{def:crac}.
\end{remark}
\begin{definition}[{\glspl{cc}}]
	An $(n,M,\lambda)$ \gls{cc} $\mathcal C$ uses an extended sender alphabet $\mathcal Q=\{\rho_{SR}:\rho_R\in\mathcal S\}$, where again $\mathcal S=\{S_N(\alpha):N+|\alpha|^2\leq E\}$. Further, it allows the sender to apply homodyne measurements on the $S$ part of the system, and change code parameters on the system in reaction to measurement results. A (causal) code is thus described by a set of functions $f_{i}:\mathbb [M]\times R^{i}\to\mathcal Q$ where $i=0,\ldots,n-1$, together with pairwise disjoint decoding sets $D_m\subset \mathbb R^n$, $m=1,\ldots,M$. The success probability of a \gls{cc} is defined as 
	\begin{align}
		\mathrm{succ} =\min_{\mathbf j^n\in\mathcal J^n} \frac{1}{M}\sum_{m=1}^M\int_{D_m}\int_{r^n}\prod_{i=1}^np(x_i,r_i|m,r^{i-1})dr^ndx^n
	\end{align}
	where the conditional probability densities are defined as
	\begin{align}
		p(x_i,r_i|m,r^{i-1}) = \Tr(H_{x_i}\mathcal N([(\mathbbm1\otimes H_{r_i})S(f(m,r^{i-1}))(\mathbbm1\otimes H_{r_i})]\otimes S(\mathbf j_i))).
	\end{align}
	A non-negative number $R$ is called achievable with $\glspl{cc}$ if there exists a sequence of $\glspl{cc}$ such that $\liminf_{n\to\infty}\tfrac{1}{n}\log M_n\geq R$.
	The corresponding capacity $C_Q$ is the supremum over all achievable numbers $R$.
\end{definition}
\begin{remark}
	Each function $f_i^{(m)}$ lets the sender, based on the previous $i$ measurement results, decide which state to create next in order to transmit the message $m$.
\end{remark}
\begin{remark}
	One prominent extension of the states $S_N(\alpha)$ are the \gls{tmsv} states
	\begin{align}\label{def:tmsv}
		|\gls{tmsv}\rangle =\cosh(2r)^{-1/2}\sum_{n=0}^\infty\tanh(r)^n|n\rangle_A|n\rangle_B,
	\end{align}
	where $r$ is the squeezing parameter and $\tr_B(|\gls{tmsv}\rangle\langle\gls{tmsv}|)=S_{\cosh(r)^2-1}(0)$. These states are entangled. A second prominent example is a thermal state $S_{2N}$ that got mixed on a $50:50$ beam-splitter, so that the resulting state is 
	\begin{align}\label{def:S2N}
		S_N^{(2)}=\tfrac{1}{\pi\cdot 2N} \int \exp(-|\alpha|^2/2N)|\tfrac{\alpha}{\sqrt{2}}\rangle\langle\tfrac{\alpha}{\sqrt{2}}|\otimes|\tfrac{\alpha}{\sqrt{2}}\rangle\langle\tfrac{\alpha}{\sqrt{2}}|d\alpha.
	\end{align}
	Again, the marginal of this state equals $S_N(0)$. This second class of states is \emph{classically} correlated, but the two modes $A$ and $B$ are not entangled.
\end{remark}
Based on this definition we can state our main theorem.

\begin{theorem}\label{thm:main}
	The following are true:
	\begin{enumerate}
		\item\label{C=0} If $P\geq E$ then $C=0$
		\item\label{c>0} If $E<P$ then $C>0$
		\item\label{barC>c} For all $E,P>0$ we have $\bar C > 0$
		\item\label{C_Q>0} For all $E,P>0$ we have $C_Q>0$.
	\end{enumerate}  
\end{theorem}
\begin{remark}
	We prove property \ref{C_Q>0} based on the correlations in a \gls{tmsv} state \eqref{def:tmsv}. We note however that the exact same strategy of proof works for a state of the form $S^{(2)}_N$ as in \eqref{def:S2N}. Thus, \emph{classical} correlations are fully sufficient to show $C_Q>0$, although different members of $\mathcal Q$ will obviously lead to different lower bounds on $C_Q$. 
\end{remark}
\begin{remark}
	Since the exact point in time where the sender-side homodyne measurement is made does not get clarified by our definition of \gls{cc}, a violation of the peak power constraint could occur in principle if the sender-side measurement is carried out prior to transmission. Since on the other hand side the ability to create high energy states reliably and on purpose differs from the ability to benefit from their sporadic occurrence, our results may nonetheless shed some light on the potential use of correlations in data transmission.
\end{remark}
\begin{proof}[Proof of property \eqref{C=0}:]
	We follow the arguments in \cite{csiszarNarayan}. Assume $P\geq E$. For every $(n,\lambda)$ \gls{dc} $\mathcal C$ with signals $\mathbf s^n=(\alpha^n_m,N_m^n)$ we consider corresponding jamming sequences $\mathbf j^n=\mathbf s^n$. Then for every decoder we have 
	\begin{align}
		\mathrm{err}(\mathcal C)
		&=\max_{\mathbf j^n}\frac{1}{M}\sum_m\Tr((\mathbbm1-D_m)\mathcal N(\mathbf j^n,\mathbf s^n_m))\\
		&\geq\frac{1}{M^2}\sum_{l,m=1}^M\Tr((\mathbbm1-D_m)\mathcal N(\mathbf j^n_l,\mathbf s^n_m))\\
		&=\frac{1}{M^2}\sum_{l,m=1}^M\Tr((\mathbbm1-D_m)\mathcal N(\mathbf s^n_m,\mathbf j^n_l)),
	\end{align}
	where we switched from a maximum- to an average jamming error and then used symmetry of the $50:50$ beam-splitter. We proceed by replacing each term $\eins-D_m$ with $D_l$ for an $l\neq m$:
	\begin{align}
		\mathrm{err}(\mathcal C)&\geq\frac{1}{M^2}\sum_{l\neq m=1}^M\Tr(D_l\mathcal N(\mathbf s^n_m,\mathbf j^n_l)) + \frac{1}{M^2}\sum_{l=1}^M\Tr((\eins-D_l)\mathcal N(\mathbf s^n_l,\mathbf j^n_l))\\
		&\geq \frac{1}{M^2}\sum_{l\neq m=1}^M\Tr(D_l\mathcal N(\mathbf s^n_m,\mathbf j^n_l))\\
		&= \frac{1}{M^2}\sum_{l, m=1}^M\Tr(D_l\mathcal N(\mathbf s^n_m,\mathbf j^n_l))-\frac{1}{M^2}\sum_{l=1}^M\Tr(D_l\mathcal N(\mathbf s^n_l,\mathbf j^n_l))\\
		&\geq \min_{\mathbf s^n\in\mathcal J^n}\frac{1}{M}\sum_{l=1}^M\Tr(D_l\mathcal N(\mathbf s^n,\mathbf j^n_l))-\frac{1}{M}\\
		&=1-\mathrm{err}(\mathcal C)-\tfrac{1}{M},
	\end{align}
	which implies (for all $M\geq2$)
	\begin{align}
		\mathrm{err}(\mathcal C)\geq\tfrac{1}{2}-\tfrac{1}{2M}\geq\tfrac{1}{4}.
	\end{align}
	Therefore $P\geq E$ implies that $\mathrm{err}(\mathcal C_n)\to0$ can only be satisfied if $|M_n|\to1$ asymptotically, which means not even a single bit can be transmitted $C=0$ and therefore \eqref{C=0} is proven. 
	\end{proof}
	\begin{proof}[Proof of property \eqref{c>0}] Assume the transmitter sends a coherent state $|\alpha\rangle$ and the jammer a thermal state $S_N(\beta)$. These are represented by first and second moments 
	\begin{align}
		(\alpha,\tfrac{1}{2}\mathbbm1),\qquad (\beta,(N+\tfrac{1}{2})\mathbbm{1}).
	\end{align}
	Upon mixing the two at a $50:50$ beam-splitter and tracing out the environment, we obtain for first and second moments of the state at the receiver 
	\begin{align}
		\left( (\alpha+\mathbbm{i}\beta)/\sqrt{2},(N+1)\mathbbm1\right).
	\end{align}
	If the receiver now uses a homodyne detector, the probability of measuring result $x\in\mathbb R$ is without loss of generality given by 
	\begin{align}
		p(x) = \tfrac{1}{\sqrt{\pi(N+1)}}\exp(-\tfrac{(x-(\alpha-\beta))^2}{N+1}),
	\end{align}
	for real numbers $(\beta,N)$ satisfying $|\beta|^2+N\leq P$. Upon setting $\alpha=\sqrt{E}$, this constraint ensures that 
	\begin{align}
		\alpha-\beta &= \sqrt{E} - \beta\\
		&\geq \sqrt{E}-\sqrt{P}\\
		&>0.
	\end{align}
	It follows that, independent of the choice of Gaussian state made by the jammer, 
	\begin{align}
		\mathbb P(x>0|\alpha,\mathbf j)&=\tfrac{1}{2}(1+\mathrm{erf}\big(\tfrac{\alpha-\beta}{\sqrt{N+1}}\big))\\
		&\geq\tfrac{1}{2}(1+\min_{N'\leq P}\mathrm{erf}\big(\tfrac{\sqrt{E}-\sqrt{P-N'}}{\sqrt{N'+1}}\big))\\
		&=\tfrac{1}{2}(1+\mathrm{erf}\big(\min_{N'\leq P}\tfrac{\sqrt{E}-\sqrt{P-N'}}{\sqrt{N'+1}}\big))\\
		&\geq\tfrac{1}{2}(1+\mathrm{erf}\big(\tfrac{\sqrt{E}-\sqrt{P}}{\sqrt{P+1}}\big))\\
		&\geq\tfrac{1}{2}+\tfrac{1}{4}\mathrm{erf}\big(\tfrac{\sqrt{E}-\sqrt{P}}{\sqrt{P+1}}\big).
	\end{align}
	The same reasoning applies to $p(x<0|-\alpha,\mathbf j)$ if the sender transmits $-\alpha$. Over $n$ consecutive uses of the channel, the sender and receiver transmit effectively over a series of binary channels with transition probabilities $w_{\mathbf j_1}(y|x),\ldots,w_{\mathbf j_n}(y|x)$ ($x,y\in\{-1,1\}$) with the property 
	\begin{align}
		w_{\mathbf j_i}(x|x)\geq\tfrac{1}{2}+\nu(E,P)\qquad x\in\{-1,1\}
	\end{align}
	where $\nu(E,P):=\tfrac{1}{4}\mathrm{erf}\big((\sqrt{E}-\sqrt{P})/\sqrt{P+1}\big)$ holding uniformly over all $i=1,\ldots,n$. Since both in- and output system of this channel are binary, Theorem 1 of \cite{ahlswedeWolfowitz} applies and the capacity of $\mathcal N$ is lower bounded as	\begin{align}
		\bar C \geq \max_p\min_{v\in\Xi} I(p;v)
	\end{align}
	where $I$ is the mutual information and $\Xi$ defined as the convex hull of all those channels which are induced by legitimate jammer attacks: 
	\begin{align}
		\Xi:=\mathrm{conv}(\{w_{(\beta,N)}\}_{|\beta|^2+N\leq P}).
	\end{align}
	By assuming a symmetric distribution $\pi(-1)=\pi(1)$ of the signals $\alpha$ and $\-\alpha$ we obtain via Pinsker's inequality and Lemma \ref{lem:lipschitz}
	\begin{align}
		\bar C &\geq \min_{v\in\Xi}I(\pi,v)\\
			&\geq \min_{v\in\Xi}\tfrac{1}{2}\| v_\pi-\pi\otimes v(\pi)\|_1^2\\
			&\geq\min_{v\in\Xi}\tfrac{1}{2}|v(-1|-1)v(1|1)-v(-1|1)v(1|-1)|\\
			&\geq\min_{v\in\Xi}\tfrac{1}{2}|v(-1|-1)v(1|1)-(1-v(1|1))(1-v(-1|-1))|\\
			&=\min_{v\in\Xi}\tfrac{1}{2}|1 - v(-1|-1) - v(1|1)|\\
			&=\min_{v\in\Xi}\tfrac{1}{2}(v(-1|-1) + v(1|1) -1)\\
			&\geq \nu(E,P)\\
			&>0.
	\end{align}
	\end{proof}
	\begin{proof}[Proof of property \eqref{barC>c}] We approach the problem by assuming an amount of $n$ bits of \gls{cr} is available in $n$ transmissions. In this case, sender and receiver can use a permutation of the sender- and receiver symbol, which can be realized by the unitary $U=\exp(-\mathbbm{i}\hat a^\dagger\hat a)$. Then, for every input state $\alpha$ of the sender and jammer state $\sigma=(\beta,(N+\tfrac{1}{2})\mathbbm1)$ the rotated state is $\sigma'=(-\beta,(N+\tfrac{1}{2})\mathbbm1)$ and 
	\begin{align}
		\rho = \tfrac{1}{2}(\sigma+\sigma')
	\end{align}
	describes the effective attack of the jammer: 
	\begin{align}
		\mathcal N'(\alpha,\sigma) &= \tfrac{1}{2}(\mathcal N'(\alpha,\sigma) + \mathcal N'(-\alpha,\sigma))\\
		&= \mathcal N'(\alpha,\rho).
	\end{align}
	Upon applying a homodyne measurement at the receiver, we get (by linearity) for every signal $\mathbf s =(\alpha,0)$ and jamming signal $\mathbf j=(\beta,N)$
	\begin{align}\label{eqn:symmetric-p}
		p(x|\mathbf s,\mathbf j) = \tfrac{1}{\sqrt{\pi(N+1)}}\exp(-\tfrac{(x+\alpha)^2+\beta^2}{N+1})\cosh(\tfrac{2\beta(x+\alpha)}{N+1}).
	\end{align}
	This Gaussian mixture is shifted by the parameter $\alpha$ and independent of the sign of $\beta$. For this distribution, and with $\mathbf s=(\sqrt{E},0)$ and arbitrary $\mathbf j=(\beta,N)$ it holds
	\begin{align}
		\mathbb P(X>0|\mathbf s,\mathbf j)&=\left(2+\mathrm{erf}(\tfrac{\sqrt{E}+ \beta}{\sqrt{1+N}})+\mathrm{erf}(\tfrac{\sqrt{E}-\beta}{\sqrt{1+N}})\right)/4\\
		&=\left(2+\tfrac{2}{\sqrt{\pi}}\int_{0}^{a+b}e^{-t^2}dt + \int_{0}^{a-b}e^{-t^2}dt\right)/4\\
		&=\left(2+\tfrac{2}{\sqrt{\pi}}\int_{b-a}^{b+a}e^{-t^2}dt \right)/4\\
		&\geq\left(2+\tfrac{2}{\sqrt{\pi}}\int_{b-a}^{b+a}e^{-(b+a)^2}dt \right)/4\\
		&=\left(2+\tfrac{4a}{\sqrt{\pi}}e^{-(b+a)^2}\right)/4
	\end{align}
	where we used $a:=\sqrt{E/(1+N)}$ and $b:=\beta/\sqrt{1+N}$. In order to derive a lower bound in terms of the numbers $E,P$ we switch back to using those terms:
	\begin{align}
		\min_{\mathbf j\in\mathcal J}\mathbb P(X>0|\mathbf s,\mathbf j)&=\left(2+\tfrac{4\sqrt{E}}{\sqrt{\pi(1+N)}}\exp(-\big(\tfrac{\sqrt{E}-\beta}{\sqrt{1+N}}\big)^2)\right)/4\\
		&\geq\left(2+\tfrac{4\sqrt{E}}{\sqrt{\pi(1+P)}}\exp(-(\sqrt{E}-\sqrt{P})^2)\right)/4.
	\end{align}
	Here we maximized the term in front of the exponential separately from the exponential itself, and used the definition of $\mathcal J$. Upon using 
	\begin{align}\label{def:epsilon(E,P)}
		\epsilon(E,P):=\tfrac{\sqrt{E}}{\sqrt{\pi(1+P)}}\exp(-(\sqrt{E}-\sqrt{P})^2)
	\end{align}
	and noting that in the case the transmitted signal was $\mathbf s'=(-\sqrt{E},0)$ the same estimate 
	\begin{align}
		\min_{\mathbf j\in\mathcal J}\mathbb P(X<0|\mathbf s',\mathbf j)\geq\tfrac{1}{2}+\epsilon(E,P)
	\end{align} 
	holds for the other signal. In addition, the distribution \eqref{eqn:symmetric-p} is \emph{symmetric} with respect to joint reflection of $x$ and the first moment of $\mathbf j$ at the origin:
	\begin{align}
		p(x|(\alpha,0),\mathbf j) &= \tfrac{1}{\sqrt{\pi(N+1)}}\exp(-\tfrac{(x+\alpha)^2+\beta^2}{N+1})\cosh(\tfrac{2\beta(x+\alpha)}{N+1})\\
		&=\tfrac{1}{\sqrt{\pi(N+1)}}\exp(-\tfrac{(-x-\alpha)^2+\beta^2}{N+1})\cosh(\tfrac{2\beta(-x-\alpha)}{N+1})\\
		&=p(-x|(-\alpha,0),\mathbf j).
	\end{align}  
	We thus proceed to applying a threshold decoder deciding that $b=0$ was sent if the measured value $x$ satisfies $x>0$ has probability of receiving the bit correctly strictly greater than (although, depending on the exact values of $P$ and $E$ very close to) $1/2$. The same estimate shows that the probability of decoding $b=1$ correctly is strictly larger than $1/2$. By the symmetry property of $p(x|\mathbf s,\mathbf j)$ we know that, for every choice $\mathbf j$ of the jammer, the resulting channel from $b$ to the value decoded by the threshold decoder is a binary symmetric channel. 
	Since by \cite[Theorem 1]{ahlswedeWolfowitz} the capacity of an arbitrarily varying binary symmetric channel is lower bounded by $1-h(\epsilon)$ where $\epsilon$ is the smallest allowed bit-flip probability, we have shown that 
	\begin{align}
		\bar C&\geq 1-h(\tfrac{1}{2}+\epsilon(E,P))>0.
	\end{align}
	\end{proof}
	\begin{proof}[Proof of property \eqref{C_Q>0}] We employ a \gls{tmsv} state parameterized by its squeezing parameter $r$, which has first moments equal to zero and second moments equal to 
	\begin{align}
		V_S=&\tfrac{1}{2}\left(\begin{array}{ll}\cosh(2r)\mathbbm{1}&\sinh(2r)\mathbbm{Z}\\
			\sinh(2r)\mathbbm{Z}&\cosh(2r)\mathbbm{1}\end{array}\right)\\
		&\mathbbm{Z}=\left(\begin{array}{ll}1&0\\
			0&-1\end{array}\right),\qquad \mathbbm{1}=\left(\begin{array}{ll}1&0\\
			0&1\end{array}\right).
	\end{align}
	The joint system describing the jamming attack and the \gls{tmsv} state is $S_N(\beta)\otimes|\gls{tmsv}\rangle\langle\gls{tmsv}|$, with first- and second moments given by 
	\begin{align}
		b\oplus0\oplus0,\ \ \left(\begin{array}{lll}
			(N+\tfrac{1}{2})\mathbbm1&   0 &    0\\
			0&\cosh(2r)\mathbbm{1}&\sinh(2r)\mathbbm{Z}\\
			0&\sinh(2r)\mathbbm{Z}&\cosh(2r)\mathbbm{1}\end{array}\right),
	\end{align}
	where $b=\sqrt{2}(\Re(\beta),\Im(\beta))^T$. The $50:50$ beam-splitter acting on the first two blocks is modeled as $B\oplus\mathbbm1$ \cite{serafiniBOOK}, where 
	\begin{align}
		B=\frac{1}{\sqrt{2}}\left(\begin{array}{ll}\mathbbm1&\mathbbm1\\ -\mathbbm1&\mathbbm1\end{array}\right).
	\end{align}
	It changes the first moments to $(\Re(\beta),\Im(\beta),-\Re(\beta),\Im(\beta),0,0)^T$ and the joint covariance matrix to 
	\begin{align}
		\left(\begin{array}{lll}
			(A+B)\cdot\mathbbm1  &   (A-B)\cdot\mathbbm1  &    C\cdot\mathbbm Z\\
			(A-B)\cdot\mathbbm1  & (A+B)\cdot\mathbbm1     & - C\cdot\mathbbm Z\\
			C\cdot\mathbb Z&-C\cdot\mathbb Z& 2B\cdot\mathbbm 1\end{array}\right),
	\end{align}
	where $A=(N+\tfrac{1}{2})/2$, $B=\cosh(2r)/4$ and $C=\sinh(2r)/4$. Upon homodyning the $\hat x$ quadrature of first and third mode, the distribution of measurement results $(x_1,x_3)$ is Gaussian \cite[Eq. (5.128)]{serafiniBOOK} with covariance- and inverse covariance matrix given by
	\begin{align}
		V_G=\left(\begin{array}{ll}
			A+B&   C\\
			C&2B\end{array}\right),\qquad V_G^{-1}=\frac{1}{\mathrm{det}(V_G)}\left(\begin{array}{ll}
			2B&   -C\\
			-C&A+B\end{array}\right).
	\end{align}
	The correlation of $(X_1,X_3)$ is given by 
	\begin{align}
		\mathrm{Corr}(X_1,X_3) &:= \frac{\mathrm{Cov}(X_1,X_3)}{\sqrt{\mathrm{Var}(X_1)\mathrm{Var}(X_3)}}\\
		&= \frac{C}{\sqrt{(A+B)2B}}
	\end{align}
	and the mean by $\Re(\beta)$, which for simplicity and without loss of generality we assume to be equal to $\beta$ itself. Upon defining new random variables 
	\begin{align}\label{def:binary-correlations}
		(b_A,b_B)=\mathrm{dec}(X_1,X_3)
	\end{align}
	where $\mathrm{dec}(x_1,x_3)=((-1)^{\mathrm{sgn}(x_1)},(-1)^{\mathrm{sgn}(x_3)})$ the probabilities of the binary detection events $(b_A,b_B)\in\{-1,1\}^2$ translate into probabilities 
	\begin{align}
		\mathbb P(\Box_{\cdot}),
	\end{align}
	where $\Box_{\nearrow}$, $\Box_{\nwarrow}$,  $\Box_{\searrow}$ and $\Box_{\swarrow}$ denote the top right, top left, bottom right and bottom left quadrant of the complex plane $\mathbb C$. These probabilities have a specific geometric structure, which we carve out by using the bi-variate normal cumulative distribution function 
	\begin{align}\label{def:bivariate-normal-cumulative-distribution-function}
		\Phi_2(h,k,\rho):=\int_{-\infty}^h\int_{-\infty}^k\left[\exp\left(-\frac{x^2-2\rho xy + y^2}{2(1-\rho^2)}\right)/2\pi\sqrt{1-\rho^2}\right]dxdy.
	\end{align}
	Upon defining new random variables $X_A=(X_1-\beta)/\sqrt{A+B}$ and $X_B=X_3/\sqrt{2B}$ we obtain
	\begin{align}\label{psw}
		\mathbb P(\Box_{\swarrow}) &= \int_{-\infty}^0\int_{-\infty}^0\frac{1}{2\pi\sqrt{\mathrm{det}(V_G)}}\exp\left(-\frac{2B(x_1-\beta)^2-2C(x_1-\beta)x_3+(A+B)x_3^2}{2\mathrm{det}(V_G)}\right)dx_1dx_3\\
		&=\int_{-\infty}^0\int_{-\infty}^{-\frac{\beta}{\sqrt{A+B}}}\frac{\sqrt{2B(A+B)}}{2\pi\sqrt{\mathrm{det}(V_G)}}\exp\left(-\frac{(x_A^2+x_B)^2-x_Ax_B2C/\sqrt{2B(A+B)}}{2\mathrm{det}(V_G)/(2B(A+B))}\right)dx_Adx_B
	\end{align}
	where $\mathrm{det}(V_G)=(A+B)2B-C^2$. Upon setting 
	\begin{align}\label{def:rho}
		\rho=C/\sqrt{2B(A+B)}
	\end{align}
	we obtain 
	\begin{align}
		\sqrt{1-\rho^2} &= \sqrt{1-C^2/(2B(A+B))}\\
				&= \frac{1}{\sqrt{2B(A+B)}}\sqrt{2B(A+B)-C^2}\\
				&=\frac{1}{\sqrt{2B(A+B)}}\sqrt{\mathrm{det}(V_G)}
	\end{align}
	and can therefore simplify the above integral to 
	\begin{align}
		\mathbb P(\Box_{\swarrow})&=\int_{-\infty}^0\int_{-\infty}^{-\frac{\beta}{\sqrt{A+B}}}\frac{1}{2\pi\sqrt{1-\rho^2}}\exp\left(-\frac{(x_A^2+x_B)^2-x_Ax_B2\rho}{2(1-\rho^2)}\right)dx_Adx_B\label{eqn:density}\\
		&=\Phi_2(-\tfrac{\beta}{\sqrt{a}},0,\rho).
	\end{align}
	By the same methodology, and since each marginal of the density in \eqref{eqn:density} is a mean zero normal distribution with variance one, it further holds 
	\begin{align}\label{pnw}
		\mathbb P(\Box_{\nwarrow}) &= \mathbb P(X_B<0) - \mathbb P(\Box_{\swarrow})\\
		&=\phi(-\tfrac{\beta}{\sqrt{a}}) - \Phi_2(-\tfrac{\beta}{\sqrt{a}},0,\tfrac{C}{\sqrt{2(A+B)B}})\\
		&=\phi(-\tfrac{\beta}{\sqrt{a}}) - \mathbb P(\Box_{\swarrow}),
	\end{align}
	where 
	\begin{align}\label{def:phi}
		\phi(x):=(2\pi)^{-1/2}\int_{-\infty}^x\exp(-t^2/2)dt
	\end{align}
	is the standard normal cumulative distribution function. Further we obtain 
	\begin{align}\label{pse}
		\mathbb P(\Box_{\searrow}) &= \mathbb P(X_A<0) - \mathbb P(\Box_{\swarrow})\\
		&=\phi(0) - \Phi_2(-\tfrac{\beta}{\sqrt{a}},0,\tfrac{C}{\sqrt{2(A+B)B}})\\
		&=\tfrac{1}{2} - \mathbb P(\Box_{\swarrow})
	\end{align}
	and, due to normalization, 
	\begin{align}
		\mathbb P(\Box_{\nearrow}) &= 1 - \mathbb P(\Box_{\swarrow}) - \mathbb P(\Box_{\nwarrow}) - \mathbb P(\Box_{\searrow})\\
		&= 1 - \mathbb P(\Box_{\swarrow}) - \mathbb P(\Box_{\nwarrow}) - \mathbb P(\Box_{\searrow})\\
		&= 1 - \phi(-\tfrac{\beta}{\sqrt{a}}) - \mathbb P(\Box_{\searrow})\\
		&= \tfrac{1}{2} + \mathbb P(\Box_{\swarrow}) - \phi(-\tfrac{\beta}{\sqrt{a}}).
	\end{align}
	It is however the case that 
	\begin{align}
		\Phi_2(-x,0,\rho)&=\mathbb P(X_A<-x,X_B<0)\label{eqn:phi-to-P}\\
		&\leq\mathbb P(X_B<0)\\
		&\leq\tfrac{1}{2}.
	\end{align}
	Thus for $|\rho|\neq1$ we have $\mathbb P(\Box_{\searrow})<\tfrac{1}{2}$. In addition, it is clear from \eqref{eqn:phi-to-P} that $x\to\Phi_2(-x,0,\rho)$ is continuous and monotone increasing. Thus for any $P\geq0$ it follows that for all $|\rho|<1$ there exists a $\delta>0$ such that $\Phi_2(-x,0,\rho)<\tfrac{1}{2}-\delta$. Therefore it is reasonable to assume that, for every jamming signal $\mathbf j\in\mathcal J$, the distribution $1$ remains correlated. The exact geometric properties of the measurement results at sender and receiver will be discussed in the following paragraph ``\nameref{para:geometry}'', and the implications for the channel in paragraph ``\nameref{para:effectiveChannel}''.
	
	\paragraph{Geometric Properties}\label{para:geometry}
	We define $\qa$ to be the set of all distributions of $(B_A,B_B)=\mathrm{dec}(X_A,X_B)$ as defined in \eqref{def:binary-correlations}, with arbitrary unbounded values of $(\beta,N)$ and let $\qa(P)$ be the subset of $\qa$ for which the jammer has jamming energy per symbol at most $P$. We define further the distributions $q_c$, $q_{-1}$ and $q_1$ as
	\begin{align}
		q_c(u,v)&:=\tfrac{1}{2}\delta(u,v)\\
		q_{-1}&:=\delta_{-1}\otimes\pi\\
		q_1&:=\delta_1\otimes\pi,
	\end{align}
	the set $\Delta:=\mathrm{conv}(\{q_c,q_{-1},q_1\})$ and the set $\Delta_\delta:=\mathrm{conv}(\{q_c,\tilde q_{-1},\tilde q_1\})$ with $\tilde q_i:=(1-\delta)q_i+\delta q_c$. The energy limit $P$ of the jammer bounds the displacement $\beta$ and effect of the jammer on the covariance matrix via $N$, so that for every $P>0$ and all $(\beta,N)$ satisfying $|\beta|^2+N\leq P$ there is a $\delta(E,P)>0$ such that $\qa(P)\subset\Delta_{\delta(E,P)}$. 
	\begin{proof}
		We first show $\qa\subset\Delta$. Each member $q$ of $\qa$ is given by the formulas \eqref{psw}, \eqref{pnw} and \eqref{pse}. in order to clearly point out the step from a distribution on signals to one on bits, we write explicitly
		\begin{align}   
			q(-1,-1) &= \mathbb P(\Box_\swarrow)=:q\\
			q(-1,1) &= \mathbb P(\Box_\nwarrow)=s-q\\
			q(1,-1) &= \mathbb P(\Box_\searrow)=\tfrac{1}{2}-q\\
			q(1,1) &= \mathbb P(\Box_\nearrow)=\tfrac{1}{2}-s+q,
		\end{align}
		where $s=\phi(-\tfrac{\beta}{\sqrt{a}})$.
		We then embed $\mathcal P(\{-1,1\})$ and thereby $\Delta$ into $\mathbb R^3$ by mapping any distribution $r$ to $r(-1,-1)e_0+r(-1,1)e_1+r(1,-1)e_2$. Then 
		\begin{align}
			q_c   &= (\tfrac{1}{2},0,0)^T\\
			q_{-1}&= (\tfrac{1}{2},\tfrac{1}{2},0)^T\\
			q_{1} &= (0,0,\tfrac{1}{2})^T.
		\end{align}
		A normal to $\Delta$ is given by $\mathbf n=(1,1,0)^T$: 
		\begin{align}
			\langle q_c - q_c,\mathbf n\rangle&=0\\
			\langle q_c - q_{-1}, \mathbf n\rangle&=\langle (0,-\tfrac{1}{2},0)^T, (1,0,1)^T\rangle = 0\\
			\langle q_c - q_1,\mathbf n\rangle&=\langle(\tfrac{1}{2},0,-\tfrac{1}{2}),(1,0,1)^T\rangle =0.
		\end{align}
		Each distribution $q$ is a subset of $\qa$. To see this, we take a distribution $q$, and compute
		\begin{align}
			\langle q-q_c,\mathbf n\rangle = (q-\tfrac{1}{2}) + (\tfrac{1}{2}-q) = 0.
		\end{align}
		Therefore, $\qa\subset\Delta$. To derive a more explicit value $\delta(E,P)$ we invoke Lemma \ref{lem:lipschitz} again: For every product distribution $v\otimes w\in\mathcal P(\{-1,1\})$ we have
		\begin{align}
			\|q-v\otimes w\|_1
				&\geq|\det(q)|\\
				&= |q\cdot (\tfrac{1}{2}-s+q) - (s-q)\cdot(\tfrac{1}{2}-q)|\\
				&=|\tfrac{1}{2}q-qs+q^2-\tfrac{1}{2}s+sq+\tfrac{1}{2}q-q^2|\\
				&=|q-\tfrac{1}{2}s|\\
				&=|\phi_2(-\tfrac{\beta}{\sqrt{A+B}},0,\rho)-\phi_2(-\tfrac{\beta}{\sqrt{a}},0,0)|\\ 	
				&\geq\min_{\mathbf{j}\in\mathcal J}|\rho|\tfrac{1}{2\pi}\exp(-\tfrac{|\beta|^2}{2(A+B)(1-\rho^2)})
		\end{align}
		by Lemma \ref{lem:placket-bound} where according to \eqref{def:rho} the value $\rho$ is a function of the jamming signal $\mathbf j$, as are $A$ and $B$. We proceed with a crude lower bound as follows: For all $\mathbf j=(\beta,N)\in\mathcal J$ we have the following three inequalities:
		\begin{align}
			|\rho|&=\frac{|\sinh(2r)|}{\sqrt{\tfrac{1}{2}\cosh(2r)(\tfrac{1}{2}N+\tfrac{1}{4}\cosh(2r)+\tfrac{1}{4})}}\\
			&\geq \frac{|\sinh(2r)|}{\sqrt{\cosh(2r)(P+\cosh(2r)+1)}}\\
			A+B&=\tfrac{1}{2}N+\tfrac{1}{4}(\cosh(2r)+1)\\
				&\leq \tfrac{1}{2}P+\tfrac{1}{4}(\cosh(2r)+1)\\
			|\beta|^2&\leq P.
		\end{align}
		By noting that $\cosh(2r)=1+2E$ as well as (since the equality $\cosh(2r)^2=1-\sinh(2r)^2$ holds)  also $|\sinh(2r)|=2\sqrt{E(E+2)}$, we have shown the existence of $\delta(E,P)$ such that 
		\begin{align}
			\|q-v\otimes w\|_1 \geq\delta(E,P)
		\end{align}
		for all product distributions $v\otimes w$ and every choice of the jammer. In the representation 
		\begin{align}
			q = \lambda_c q_c + \lambda_{-1}q_{-1} + \lambda_1 q_1
		\end{align}
		this implies that for all $\mathbf j\in\mathcal J$ we have
		\begin{align}\label{eqn:delta(E,P)}
			\lambda_c\geq\delta(E,P).
		\end{align}
		
	\end{proof}
	\paragraph{Effective Channel}\label{para:effectiveChannel}
	Consider any channel modeled by a conditional probability distribution $w(y|x)$ where $x,y\in\{-1,1\}$. Assuming the sender and receiver share perfectly correlated uniformly random bits, they may transform $w$ to $w_c$ as follows:
	\begin{align}
		w_c(y|x):=\tfrac{1}{2}\sum_u w(y\oplus u|x\oplus u). 
	\end{align}
	In this case, we have 
	\begin{align}
		w_c(-1|-1)&=\tfrac{1}{2}\sum_u w(-1\oplus u|-1\oplus u)\\
		&= \tfrac{1}{2}\sum_u w(1\oplus u|1\oplus u)\\
		&=w_c(1|1)
	\end{align}
	and therefore $w_c$ is a \gls{bsc}. If instead the sender and the receiver use any uncorrelated distributions $v\otimes \pi$ with the same modulo two addition, then we get for every $x,y$
	\begin{align}
		w_u(y|x)&=\sum_{x'}v(x')\sum_{y'}\tfrac{1}{2}w(y\oplus y'|x\oplus x')\\
			&=w_u(y\oplus1|x).
	\end{align}
	Thus $w_u$ is a \gls{bsc} again, but this time one with crossover probability $1/2$. 
	The code for showing $C_Q>0$ is now straightforward: During $2n$ transmissions, the sender and receiver send a \gls{tmsv} state in the first $n$ transmissions, leaving them with a series $q_1,\ldots,q_n\in\qa$ of joint distributions. During the next $n$ transmissions, they transmit using the \gls{bpsk} alphabet $\{|-\sqrt{E}\rangle,|\sqrt{E}\rangle\}$ as in the proof of property \eqref{barC>c} of Theorem \ref{thm:main}, leading to a series $w_i(y|x)$ of channels. Afterwards, again as in the proof of property \eqref{barC>c} of Theorem \ref{thm:main}, they utilize their distributions $q_i$ to transform each $w_i$ into a \gls{bsc}:
	\begin{align}
		\tilde w_i(y|x)
			&:=\textstyle\sum_{x',y'}w_i(y\oplus y'|x\oplus x')q(x',y')\\
			&=\lambda_{c,i} \textstyle\sum_{x'}\tfrac{w_i(y\oplus y'|x\oplus x')}{2}+\textstyle\sum_{y'}[\lambda_{-1,i}\tfrac{w_i(y\oplus y'|x\oplus -1)}{2}+\lambda_{1,i}\tfrac{w_i(y\oplus y'|x\oplus 1)}{2}]\\
			&=\lambda_{c,i}\cdot b_{p_i}(y|x) + (1-\lambda_{c,i})\cdot b_{1/2}(y|x).
	\end{align}
	As follows from the proof of property \eqref{barC>c} of Theorem \ref{thm:main} it holds $p_i\leq\epsilon(E,P)$ uniformly over $i=1,\ldots,n$ for some $0\leq \epsilon(E,P)<\tfrac{1}{2}$. Further we have for every two \glspl{bsc}, $\lambda\in[0,1]$ and $x,y$ the property 
	\begin{align}
		\lambda \cdot b_{p_1}(y|x) + (1-\lambda)\cdot b_{p_2}(y|x) = b_{\lambda p_1+(1-\lambda)p_2}(y|x).
	\end{align}
	Therefore each effective channel is a \gls{bsc} with crossover probability 
	\begin{align}
		\lambda_{c,i}\cdot p_i + (1-\lambda_{c,i})\cdot \tfrac{1}{2}\leq \lambda_{c,i}\cdot \epsilon(E,P) + (1-\lambda_{c,i})\cdot \tfrac{1}{2}<\tfrac{1}{2},
	\end{align}
	where the last inequality follows from the uniform bound $\lambda_{c,i}>\delta(E,P)$ of \eqref{eqn:delta(E,P)}.

	As before in the proof of property \ref{c>0} an application of \cite[Theorem 1]{ahlswedeWolfowitz} shows that this implies $C_Q>0$. 
\end{proof}

\section{Appendix}
	The following two lemmas are important in our derivation:
	\begin{lemma}\label{lem:lipschitz}
		let $p,q\in\mathcal P(\{-1,1\}^2)$ with $p$ being a product distribution which factorizes as $p(i,j)=v(i)\cdot w(j)$ for some $v,w\in\mathcal P(\{-1,1\})$. Then
		\begin{align}
			\|p-q\|_1\geq|q(1,1)q(-1,-1)-q(1,-1)q(-1,1)|.
		\end{align}
	\end{lemma}
	\begin{proof}
		Let $\mathrm{det}(p):=p(1,1)p(-1,-1)-p(1,-1)p(-1,1)$ and $\mathrm{det}(q)$ in the same way. Then 
		\begin{align}
			\mathrm{det}(p)-\mathrm{det}(q) &= p(1,1)p(-1,-1)-p(1,-1)p(-1,1)\\
			&\qquad - q(1,1)q(-1,-1)-q(1,-1)q(-1,1)\\
			&=p(1,1)p(-1,-1)-p(1,-1)p(-1,1)  - q(1,1)q(-1,-1)\\
			&\qquad-q(1,-1)q(-1,1) + p(-1,1)q(1,-1) - p(-1,1)q(1,-1)\\
			&=p(-1,-1)[p(1,1)-q(1,1)] + q(1,1)[p(-1,-1)-q(-1,-1)]\\
			&\qquad + p(-1,1)[q(1,-1)-p(1,-1)] + q(1,-1)[q(-1,1)-p(-1,1)]\\
			&\leq \sum_{i,j}|p(i,j)-q(i,j)|\\
			&=\|p-q\|_1
		\end{align}
		Since $\mathrm{det}(v\otimes w)=0$, the lemma is proven.
	\end{proof}
	The second argument we need is the following Lemma
	\begin{lemma}\label{lem:placket-bound}
		Let $\phi_2$ be the bi-variate normal cumulative distribution function \eqref{def:bivariate-normal-cumulative-distribution-function} and $\phi$ the standard normal cumulative distribution function \eqref{def:phi}. Then for all $t\in\mathbb R$ and $\rho\in[-1,1]$ we have
		\begin{align}
			|\phi_2(t,0,\rho)-\tfrac{1}{2}\phi(t)| \geq |\rho|\tfrac{1}{2\pi}\exp(-\tfrac{t^2}{2(1-\rho^2)}).
		\end{align}
	\end{lemma}
	\begin{proof}[Proof of Lemma \ref{lem:placket-bound}]
		We start by noting that according to Plackett's formula \cite{plackett1954} (see \cite[Eq. (10)]{Komelj2023} for the exact formulation as we use it here) we have
		\begin{align}
	\partial_\rho\phi_2(t,0,\rho)=\tfrac{1}{2\pi\sqrt{1-\rho^2}}\exp\left(-\frac{t^2}{2(1-\rho^2)}\right).
	\end{align}
	Therefore we have 
	\begin{align}
		\phi_2(t,0,\rho)-\tfrac{1}{2}\phi(t) &=  \phi_2(t,0,\rho)-\phi_2(t,0,0)\\
		&=\int_0^\rho\partial_\rho\phi_2(t,0,\rho)d\rho\\
		&=\int_0^\rho \tfrac{1}{2\pi\sqrt{1-\rho^2}}\exp\left(-\frac{t^2}{2(1-\rho^2)}\right) d\rho.
	\end{align}
	We can therefore derive the lower bound 
	\begin{align}
		|\phi_2(t,0,\rho)-\tfrac{1}{2}\phi(t)|&\geq|\rho|\min_{0\leq\nu\leq\rho}\tfrac{1}{2\pi\sqrt{1-\nu^2}}\exp(-\tfrac{t^2}{2(1-\nu^2)})\\
		&\geq|\rho|\min_{0\leq\nu\leq\rho}\tfrac{1}{2\pi}\exp(-\tfrac{t^2}{2(1-\rho^2)})\\
		&\geq|\rho|\tfrac{1}{2\pi}\exp(-\tfrac{t^2}{2(1-\rho^2)}).
	\end{align}
	\end{proof}
\section*{Acknowledgements} This work was supported in part by the DFG Emmy-Noether Program under Grant 1129/2-1, 
in part by the Federal Ministry of Research, Technology and Space under grant numbers 16KISQ039, 16KISR026, 16KIS1598K, 16KISQ093, 16KISQ077, and 16KIS2604 and in part by the Federal Ministry of Research, Technology and Space of Germany through the Programme of "Souverän. Digital. Vernetzt." Joint Project 6G-life, project identification number 16KISK002.
The generous support of the state of Bavaria via the 6GQT project is greatly appreciated. This research is part of the Munich Quantum Valley, which is supported by the Bavarian state government with funds from the Hightech Agenda Bayern Plus.
\bibliographystyle{plain}
\bibliography{bib,bibcopy}

@article{Boche13,
    author = {Boche, H. and Nötzel, J.},
    title = {Arbitrarily small amounts of correlation for arbitrarily varying quantum channels},
    journal = {Journal of Mathematical Physics},
    volume = {54},
    number = {11},
    pages = {112202},
    year = {2013},
    month = {11},
    issn = {0022-2488},
    doi = {10.1063/1.4825159},
    url = {https://doi.org/10.1063/1.4825159},
    eprint = {https://pubs.aip.org/aip/jmp/article-pdf/doi/10.1063/1.4825159/15752911/112202\_1\_online.pdf},
}

@article{plackett1954,
  title={A reduction formula for normal multivariate integrals},
  author={Plackett, Robin L},
  journal={Biometrika},
  volume={41},
  number={3/4},
  pages={351--360},
  year={1954},
  publisher={JSTOR}
}

@article{minglaiSuperActivation,
    author = {Boche, Holger and Cai, Minglai and Deppe, Christian and Nötzel, Janis},
    title = {Classical-quantum arbitrarily varying wiretap channel: Secret message transmission under jamming attacks},
    journal = {Journal of Mathematical Physics},
    volume = {58},
    number = {10},
    pages = {102203},
    year = {2017},
    month = {10},
    issn = {0022-2488},
    doi = {10.1063/1.5005947},
    url = {https://doi.org/10.1063/1.5005947},
    eprint = {https://pubs.aip.org/aip/jmp/article-pdf/doi/10.1063/1.5005947/15698004/102203_1_online.pdf},
}

@article{minglaiDiscontinuity,
   author = {Boche, Holger and Cai, Minglai and Deppe, Christian and Nötzel, Janis},
   title={Classical-quantum arbitrarily varying wiretap channel: common randomness assisted code and continuity},
   volume={16},
   ISSN={1573-1332},
   url={http://dx.doi.org/10.1007/s11128-016-1473-y},
   DOI={10.1007/s11128-016-1473-y},
   number={1},
   journal={Quantum Information Processing},
   publisher={Springer Science and Business Media LLC},
   year={2016},
   month={12} }

@article{noetzelDiscontinuity,
    author = {Boche, H. and Nötzel, J.},
    title = {Positivity, discontinuity, finite resources, and nonzero error for arbitrarily varying quantum channels},
    journal = {Journal of Mathematical Physics},
    volume = {55},
    number = {12},
    pages = {122201},
    year = {2014},
    month = {12},
    issn = {0022-2488},
    doi = {10.1063/1.4902930},
    url = {https://doi.org/10.1063/1.4902930},
    eprint = {https://pubs.aip.org/aip/jmp/article-pdf/doi/10.1063/1.4902930/15974951/122201_1_online.pdf},
}

@INPROCEEDINGS{bocheUncomputable,
  author={Boche, Holger and Schaefer, Rafael F and Vincent Poor, H.},
  booktitle={ICASSP 2019 - 2019 IEEE International Conference on Acoustics, Speech and Signal Processing (ICASSP)}, 
  title={Detectability of Denial-of-service Attacks on Communication Systems}, 
  year={2019},
  volume={},
  number={},
  pages={2532-2536},
  doi={10.1109/ICASSP.2019.8683553}}

@misc{detectingsymmetrizability,
      title={Detecting Symmetrizability in Physical Systems}, 
      author={Florian Seitz and Janis Nötzel},
      year={2025},
      eprint={2512.02869},
      archivePrefix={arXiv},
      primaryClass={quant-ph},
      url={https://arxiv.org/abs/2512.02869}, 
}

@ARTICLE{ahlswede-cai-correlated,
  author={Ahlswede, R. and Ning Cai},
  journal={IEEE Transactions on Information Theory}, 
  title={Correlated sources help transmission over an arbitrarily varying channel}, 
  year={1997},
  volume={43},
  number={4},
  pages={1254-1255},
  doi={10.1109/18.605589}}

@INPROCEEDINGS{jaggiAVC,
  author={Bhattacharya, Sagnik and Budkuley, Amitalok J. and Jaggi, Sidharth},
  booktitle={2019 IEEE International Symposium on Information Theory (ISIT)}, 
  title={Shared Randomness in Arbitrarily Varying Channels}, 
  year={2019},
  volume={},
  number={},
  pages={627-631},
  doi={10.1109/ISIT.2019.8849801}}

@ARTICLE{peregAVC,
  author={Pereg, Uzi and Steinberg, Yossef},
  journal={IEEE Transactions on Information Theory}, 
  title={The Arbitrarily Varying Channel With Colored Gaussian Noise}, 
  year={2021},
  volume={67},
  number={6},
  pages={3781-3817},
  doi={10.1109/TIT.2021.3063905}}

@article{ahlswedeWolfowitz, 
    author = {Ahlswede, R. and Wolfowitz, J.},
    year = {1970},
    title = {The capacity of a channel with arbitrarily varying channel probability functions and binary output alphabet},
    journal = {Zeitschrift für Wahrscheinlichkeitstheorie und Verwandte Gebiete},
    pages = {186 --194},
    volume = {15},
    issue = {3},
    sn = {1432-2064},
    url = {https://doi.org/10.1007/BF00534915},
    doi = {10.1007/BF00534915},
}

@INPROCEEDINGS{guhaBeamsplitter-1,
  author={Guha, Saikat and Shapiro, Jeffrey H.},
  booktitle={2007 IEEE International Symposium on Information Theory}, 
  title={Classical Information Capacity of the Bosonic Broadcast Channel}, 
  year={2007},
  volume={},
  number={},
  pages={1896-1900},
  doi={10.1109/ISIT.2007.4557498}}

@article{guhaBeamsplitter-2,
  title = {Classical capacity of bosonic broadcast communication and a minimum output entropy conjecture},
  author = {Guha, Saikat and Shapiro, Jeffrey H. and Erkmen, Baris I.},
  journal = {Phys. Rev. A},
  volume = {76},
  issue = {3},
  pages = {032303},
  numpages = {12},
  year = {2007},
  month = {Sep},
  publisher = {American Physical Society},
  doi = {10.1103/PhysRevA.76.032303},
  url = {https://link.aps.org/doi/10.1103/PhysRevA.76.032303}
}

@article{Komelj2023,
	title={The Bivariate Normal Integral via Owen's Function as a Modified Euler's Arctangent Series},
	volume={13},
	ISSN={2161-1211},
	url={http://dx.doi.org/10.4236/ajcm.2023.134026},
	DOI={10.4236/ajcm.2023.134026},
	number={04},
	journal={American Journal of Computational Mathematics},
	publisher={Scientific Research Publishing, Inc.},
	author={Komelj, Janez},
	year={2023},
	pages={476–504} 
}

@inproceedings{fullyQuantumAVC,
	doi = {10.1109/isit.2018.8437610},
	url = {https://doi.org/10.1109},
	year = 2018,
	month = {jun},
	publisher = {{IEEE}},	
	author = {Holger Boche and Christian Deppe and Janis N\"otzel and Andreas Winter},
	title = {Fully Quantum Arbitrarily Varying Channels: Random Coding Capacity and Capacity Dichotomy},
	booktitle = {2018 {IEEE} International Symposium on Information Theory ({ISIT})}
}

@book{serafiniBOOK,
author = {Serafini, Alessio},
publisher = {CRC Press},
title = {{Quantum Continuous Variables: A Primer of Theoretical Methods}},
year = {2017}
}

@ARTICLE{csiszarNarayan,  
    author={Csiszar, I. and Narayan, P.},  
    journal={IEEE Transactions on Information Theory},   
    title={Capacity of the Gaussian arbitrarily varying channel},   
    year={1991},  
    volume={37},  
    number={1},  
    pages={18-26},  
    doi={10.1109/18.61125}
}

@article{bbt,
    title={The Capacities of Certain Channel Classes Under Random Coding},
    author={D. Blackwell and L. Breiman and A. J. Thomasian},
    journal={The Annals of Mathematical Statistics},
    pages={558-567},
    volume={31},
    number={3},
    year={1960},
}

@article{elimination,
    author={Ahlswede, R.},
    title={Elimination of correlation in random codes for arbitrarily varying channels},
    journal={Z. Wahrscheinlichkeitstheorie verw Gebiete},
    volume={44}, 
    pages={159 - 175},
    year={1978},
    doi={10.1007/BF00533053},
}

@article{quantumAVC,
    author = {Ahlswede, R. and Bjelakovic, I. and Boche, H. and N\"otzel, J.},
    title = {Quantum Capacity under Adversarial Quantum Noise: Arbitrarily Varying Quantum Channels},
    journal = {Communications in Mathematical Physics},
    volume = {317},
    pages = {103 - 156},
    year = {2013},
    url = {https://doi.org/10.1007/s00220-012-1613-x},
}

@ARTICLE{csiszarNarayanPositivity,
  author={Csiszar, I. and Narayan, P.},
  journal={IEEE Transactions on Information Theory}, 
  title={The capacity of the arbitrarily varying channel revisited: positivity, constraints}, 
  year={1988},
  volume={34},
  number={2},
  pages={181-193},
  doi={10.1109/18.2627}}

@ARTICLE{ahlswede-blinovsky,
  author={Ahlswede, Rudolf and Blinovsky, Vladimir},
  journal={IEEE Transactions on Information Theory}, 
  title={Classical Capacity of Classical-Quantum Arbitrarily Varying Channels}, 
  year={2007},
  volume={53},
  number={2},
  pages={526-533},
  doi={10.1109/TIT.2006.889004}}

@INPROCEEDINGS{noetzel-isit2024,
  author={Nötzel, Janis},
  booktitle={2024 IEEE International Symposium on Information Theory (ISIT)}, 
  title={Data Transmission over a Bosonic Channel under Classical Noise}, 
  year={2024},
  volume={},
  number={},
  pages={1227-1232},
  doi={10.1109/ISIT57864.2024.10619133}}

@ARTICLE{ericsonAVC,
  author={Ericson, T.},
  journal={IEEE Transactions on Information Theory}, 
  title={Exponential error bounds for random codes in the arbitrarily varying channel}, 
  year={1985},
  volume={31},
  number={1},
  pages={42-48},
  doi={10.1109/TIT.1985.1056995}}

\end{document}